\documentclass[orivec]{llncs}

\usepackage{amsmath}
\usepackage{graphicx}
\usepackage[english]{babel}
\usepackage{lmodern}
\usepackage[T1]{fontenc}
\usepackage[utf8]{inputenc}
\usepackage{multirow}
\usepackage{dcolumn}
\usepackage[tworuled,linesnumbered,noend]{algorithm2e}
\usepackage[utf8]{inputenc}

\let\doendproof\endproof
\newif\ifqed
\newcommand{\theqed}{~\hfill\qed}
\renewcommand{\endproof}{\ifqed\theqed\fi\global\qedtrue\doendproof}

\newcommand{\softO}{\widetilde{O}}

\qedtrue

\newcommand{\C}{\ensuremath{\mathcal{C}}\xspace}
\newcommand{\pspt}[1]{$#1$-\texttt{PASPT}}
\renewcommand{\epsilon}{\varepsilon}

\newcommand{\hide}[1]{\relax}
\newcommand{\tree}[2]{\ensuremath{T_{#1}(#2)}}
\newcommand{\subtree}[3]{\ensuremath{T_{#1}(#2,#3)}}

\newcommand{\erdos}{Erd\H{o}s-R\'enyi\xspace}
\newcommand{\barabasi}{Barab\'asi-Albert\xspace}

\title{Path-Fault-Tolerant Approximate\\
Shortest-Path Trees\thanks{Research partially supported by the Italian Ministry
of University and Research
under the Research Grants: 2010N5K7EB PRIN 2010 ``ARS TechnoMedia'' (Algoritmica
per le Reti Sociali Tecno-mediate), and 2012C4E3KT PRIN 2012 ``AMANDA''
(Algorithmics for MAssive and Networked DAta).}}
\author{Annalisa D'Andrea\inst{1}, Mattia D'Emidio\inst{1}, Daniele
Frigioni\inst{1}, \\ Stefano Leucci\inst{1}, Guido Proietti\inst{1,2}}
\institute{Dipartimento di Ingegneria e Scienze dell'Informazione e Matematica, 
Università degli Studi dell'Aquila, Via Vetoio, I--67100 L'Aquila, Italy.
\and
Istituto di Analisi dei Sistemi ed Informatica
``Antonio Ruberti'', Consiglio Nazionale delle Ricerche, Via dei Taurini 19,
I--00185 Roma, Italy.\\
\email{\{annalisa.dandrea, stefano.leucci\}@graduate.univaq.it\\
\{mattia.demidio, daniele.frigioni, guido.proietti\}@univaq.it}}

\begin{document}

\pagestyle{plain}	
\maketitle
\newcommand{\lcomment}[1]{ \tag{#1} }

\begin{abstract}
Let $G=(V,E)$ be an $n$-nodes non-negatively real-weighted undirected graph. In
this paper we show how to enrich a {\em single-source shortest-path tree} (SPT)
of $G$ with a \emph{sparse} set of \emph{auxiliary} edges selected from $E$, in
order to create a structure which tolerates effectively a \emph{path failure} in
the SPT. This consists of a simultaneous fault of a set $F$ of at most $f$
adjacent edges along a shortest path emanating from the source, and it is
recognized as one of the most frequent disruption in an SPT. We show that, for
any integer parameter $k \geq 1$, it is possible to provide a very sparse (i.e.,
of size $O(kn\cdot f^{1+1/k})$) auxiliary structure that carefully approximates
(i.e., within a stretch factor of $(2k-1)(2|F|+1)$) the true shortest paths from
the source during the lifetime of the failure. Moreover, we show that our
construction can be further refined to get a stretch factor of $3$ and a size of 
$O(n \log n)$ for the special case $f=2$, and that it can be converted into a 
very 
efficient \emph{approximate-distance
sensitivity oracle}, that allows to quickly (even in optimal time, if $k=1$)
reconstruct the shortest paths (w.r.t. our structure) from the source after a
path failure, thus permitting to perform promptly the needed rerouting
operations. Our structure compares favorably with previous known solutions, as
we discuss in the paper, and moreover it is also very effective in practice, as
we assess through a large set of experiments.
\end{abstract}

\section{Introduction}
Broadcasting data from a source node to every other node of a network is one of
the most basic communication primitives in modern networked applications.
Given the widespread diffusion of such applications, in the recent past, there
has been an increasing demand for more and more efficient, i.e. scalable and
reliable, methods to implement this fundamental feature.

The natural solution is that of modeling the network as a graph (nodes as
vertices and links as edges) and building a (fast and compact) structure to be
used to transmit the data. In particular, the most common approach of this kind
is that of computing a {\em shortest-path tree} (SPT), rooted at the desired
source node, of such graph.

However, the SPT, as any tree-based topology, is prone to unpredictable events
that might occur in practice, such as failures of nodes and/or links. Therefore,
the use of SPTs might result in a high sensitivity to malfunctioning, which
unavoidably causes the undesired effect of disconnecting sets of nodes from the
source and thus the interruption of the broadcasting service.

Therefore, a general approach to cope with this scenario is to make the SPT
\emph{fault-tolerant} against a given number of simultaneous component failures,
by adding to it a set of suitably selected edges from the underlying graph, so
that the resulting structure will remain connected w.r.t. the source. In other
words, the selected edges can be used to build up alternative paths from the
root, each one of them in replacement of a corresponding original shortest path
which was affected by the failure. However, if these paths are constrained to be
\emph{shortest}, then it can be easily seen that for a non-negatively real
weighted and undirected graph of $n$ nodes and $m$ edges, this may require
as much as $\Theta(m)$ additional edges, also in the case in which
$m=\Theta(n^2)$. In other words, the set-up costs of the strengthened network
may become unaffordable.

Thus, a reasonable compromise is that of building \emph{sparse} and
fault-tolerant structure which \emph{approximates} the shortest paths from the
source, i.e., that contains paths which are guaranteed to be longer than the
corresponding shortest paths by at most a given \emph{stretch} factor, for any
possible edge/vertex failure that has to be handled. In this way, the obtained
structure can be revised as a 2-level communication network: a first
\emph{primary} level, i.e., the SPT, which is used when all the components are
operational, and an \emph{auxiliary} level which comes into play as soon as a
component undergoes a failure.

In this paper, we show that an efficient structure of this sort exists for a
prominent class of failures in an SPT, namely those involving a set of adjacent
edges along a shortest path emanating from the source of the SPT. Our study is
motivated by several applications, such as, for instance, traffic engineering in
optical networks or path-congestion management in road-networks, where failures
in the above form often affect the SPT~\cite{BW09,DDFLP13,MCFF09}.
For this kind
of failure, also known as a \emph{path failure}\footnote{Notice that this is a
small abuse of nomenclature, since failures we consider are restricted to the
path's edges only.}, we show that it is possible not only to obtain resilient
sparse structures, but also that these can be pre-computed efficiently, and that
they can return quickly the auxiliary network level.

\subsection{Related Work}
In the recent past, many efforts have been dedicated to devising single and
multiple edge/vertex fault-tolerant structures. More formally, let $r$ denote a
distinguished source vertex of a non-negatively real-weighted and undirected
graph $G=(V(G),E(G))$, with $n$ nodes and $m$ edges. We say that a spanning
subgraph $H$ of $G$ is an \emph{Edge/Vertex-fault-tolerant $\alpha$-Approximate
SPT}  (in short, $\alpha$-{\ttfamily E/VASPT}), with $\alpha>1$, if it satisfies
the following condition: For each edge $e \in E(G)$ (resp., vertex $v \in
V(G)$), all the distances from $r$ in the subgraph $H-e$, i.e., $H$ deprived of
edge $e$ (resp., the subgraph $H-v$, i.e., $H$ deprived of vertex $v$ and all
its incident edges) are $\alpha$-stretched (i.e., at most $\alpha$ times longer)
w.r.t. the corresponding distances in $G-e$ (resp., $G-v$).

An early work on the matter is \cite{NPW03}, where the authors showed that by
adding at most $n-1$ edges to the SPT, a 3-{\ttfamily EASPT} can be obtained.
This was shown to be very useful in order to compute a recovery scheme needing
only one backup routing table at each node \cite{IIOY03}. In \cite{GP07}, the
authors showed instead how to build a 1-{\ttfamily EASPT} in
$\softO(m n)$ time\footnote{The $\softO$ notation hides poly-logarithmic
factors in $n$.}. Notice that, a 1-{\ttfamily EASPT} contains \emph{exact}
replacement paths from the source, but of course its size might be $\Theta(n^2)$
if $G$ is dense. Then, in \cite{BK13}, Baswana and Khanna devised a
$3$-{\ttfamily VASPT} of size $O(n \log n)$. Later on, a significant improvement
to this
result was provided in \cite{BGLP14}, where the authors showed the existence of
a $(1+\varepsilon)$-{\ttfamily E/VASPT}, for any $\varepsilon >0$, of size
$O(\frac{n \log n}{\varepsilon^2})$.

Concerning \emph{unweighted} graphs, in \cite{BK13} the authors give a
$(1+\varepsilon)$-{\ttfamily VABFS} (where BFS stands for \emph{breadth-first
search tree}) of size $O(\frac{n}{\varepsilon^3}+n \log
n)$ (actually, such a size can be easily reduced to
$O(\frac{n}{\varepsilon^3})$). Then, Parter and Peleg in \cite{PP14} present a
set of lower and upper bounds to the size of a
$(\alpha,\beta)$-{\ttfamily EABFS}, namely a structure for which the length of a
path is
stretched by at most a factor of $\alpha$, plus an additive term of $\beta$.
More
precisely, they construct a $(1,4)$-{\ttfamily EABFS} of size $O(n^{4/3})$.
Moreover, assuming at most $f=O(1)$ edge failures can take place, they show the
existence of a $(3(f +1),(f+1) \log n)$-{\ttfamily EABFS} of size $O(fn)$.
This was improving onto the general fault-tolerant \emph{spanner} construction
given
in \cite{CLPR09}, which, for weighted graphs and for any integer parameter $k
\geq 1$, is resilient to up to $f$ edge failures with stretch
factor of $2k-1$ and size $O(f \cdot n^{1+1/k})$.

\hide{Finally, we mention the general fault-tolerant \emph{spanner} construction
given
in \cite{CLPR09}, which, for any integer parameter $k \geq 1$, is resilient to
up to $f$ edge/vertex failures with stretch
factor of $(2k-1)$ and size $\softO(f^2 \cdot k^{f+1} \cdot n^{1+1/k})$. This
latter result has been then improved through a randomized construction in
\cite{DK11}, where the expected size was reduced to $\softO(f^{2-1/k} \cdot
n^{1+1/k})$.}

On the other hand, concerning \emph{approximate-distance sensitivity oracles}
(simply \emph{$\alpha$-oracles} in the
following, where $\alpha$ denotes the guaranteed approximation ratio w.r.t. true
distances), researchers aimed at computing, with a \emph{low} preprocessing
time, a \emph{compact}
data structure able to \emph{quickly} answer to some distance query following
an edge/vertex failure. The vast literature dates back to the work \cite{TZ05}
of Thorup and Zwick, who
showed that, for any integer $k \geq 1$, any undirected graph with non-negative
edge weights can be preprocessed in $O(km \cdot n^{1/k})$ time to build a
$(2k-1)$-oracle of size $O(k\cdot n^{1+1/k})$, answering in $O(k)$ time to a
post-failure distance query, recently reduced to $O(1)$ time in
\cite{DBLP:conf/stoc/Chechik14}. Due to the long-standing girth conjecture
of Erd\H{o}s \cite{Erd64}, this is essentially optimal. Concerning the failure
of a set $F$ of at most $f$ edges, in \cite{CLPR10} the authors built, for any
integer $k \geq 1$, a $(8k - 2)(f + 1)$-oracle of size $O(fk \cdot n^{1+1/k}
\log (n W))$, where $W$ is the ratio of the maximum to the minimum edge weight
in $G$, and with a query time of
$\softO(|F| \cdot \log \log d)$, where $d$ is the actual distance between the
queried pair of nodes in $G-F$. As far as \emph{SPT oracles} (i.e., returning
distances/paths only from a source node) are concerned, in \cite{BK13}
it is shown how to build in $O(m \log n + n \log^2 n)$ time an SPT oracle of
size $O(n \log n)$, that for any single-vertex-failure
returns a 3-stretched replacement path in time proportional to the path's size.
Finally, for directed graphs with integer positive edge weights bounded by $M$,
in \cite{GW12} the authors show how to
build in $\softO(M n^{\omega})$ time and $\Theta(n^2)$ space a randomized
single-edge-failure SPT oracle returning \emph{exact} distances in $O(1)$ time,
where $\omega< 2.373$ denotes the matrix
multiplication exponent.

\subsection{Our Results}
In this paper, we consider the specific, yet interesting, problem of making a
SPT resilient to the failure of any sub-path of size (i.e., number of edges) at
most $f \geq 1$ emanating from its source.

More in details, let $F$ be a set of cascading edges of a given SPT, where
$0<|F|\leq f$. We say
that a spanning subgraph $H$ of $G$ is
a \emph{Path-Fault-Tolerant $\alpha$-Approximate SPT} (in short,
$\alpha$-\texttt{PASPT}), with
$\alpha\ge 1$, if, for each vertex $z \in V(G)$, the following inequality holds:
$
d_{H-F}(z) \leq \alpha \cdot d_{G-F}(z)
$,
where $d_{G-F}(z)$ (resp., $d_{H-F}(z)$)
denotes the distance from $r$ to $z$ in $G-F$ (resp., $H-F$).
For any integer parameter $k \geq 1$, we can provide the following results:

\begin{itemize}

\item
We give an algorithm for computing, in $O(n\cdot(m +f^2))$ time, a
\pspt{(2k-1)(2|F|+1)} containing
$O(k n\cdot f^{1+\frac{1}{k}})$ edges;

\item
We give an algorithm for computing, in $O(n \cdot (m + f^2) )$ time, an oracle
of size $O(k n\cdot f^{1+\frac{1}{k}})$ which
is able to return: (i) a $(2k-1)(2|F|+1)$-approximate distance in $G-F$ between
$r$ and
a generic vertex $z$ in $O(k)$ time; (ii) the associated path in $O(k +f +
\ell)$ time, where $\ell$ is the number of its edges; if $k=1$, this can be
further reduced to $O(\ell)$ time.
\end{itemize}

Concerning the former result, it compares favorably with both the aforementioned
general fault-tolerant spanner constructions given in \cite{CLPR09}, and the
unweighted {\ttfamily EABFS} provided in \cite{PP14}, while concerning instead
the latter result, it compares favorably with the fault-tolerant oracle given in
\cite{CLPR10}.  For the sake of fairness, we remind that all these structures
were thought to cope with edge failures arbitrarily spread across $G$, though.

Besides that, we also analyze in detail the special case when at most $f = 2$
failures of cascading edges can occur,
for which we are able to achieve a significantly better stretch factor. More
precisely, we design:
(i) an algorithm for computing, in $O(n \cdot (m +n \log n))$ time, a 3-\texttt{PASPT}
containing $O(n \log n)$ edges;
(ii) an algorithm for computing, in $O(n \cdot (m +n \log n))$ time, an oracle
of size $O(n \log n)$ which is able to return a
$3$-approximate distance in $G-F$ between $r$ and a generic vertex $z$ in constant time, and
the associated path in a time proportional to the number of its edges.
Some of the proofs related to these latter results
will be given in the appendix.

Finally, we provide an experimental evaluation of the proposed
structures, to assess their performance in practice w.r.t.\ both size and
quality of the stretch. 
\section{Notation}
\label{sec:preliminaries}
In what follows, we give our notation for the considered problem.
We are given a non-negatively real-weighted, undirected graph $G=(V(G),E(G))$ with
$|V(G)|=n$ vertices and $|E(G)|=m$ edges. We denote by $w_G(e)$ or $w_G(u,v)$
the weight of the edge $e=(u,v) \in E(G)$.
Given an edge $e=(u,v)$, we denote by $G - e$ or $G - (u,v)$ the graph obtained
from $G$ by removing the edge $e$. Similarly, for a set $F$ of edges, $G-F$
denotes the graph obtained from $G$ by removing the edges in $F$. Furthermore,
given a vertex $v \in V(G)$, we denote by $G - v$ the graph obtained from $G$ by
removing vertex $v$ and all its incident edges.
Given a graph $G$, we call $\pi_G(x,y)$ a shortest path between two vertices
$x,y \in V(G)$,
$d_G(x,y)$ its weighted length (i.e., the distance from $x$ to $y$ in $G$),
$\tree{G}{r}$ a shortest path tree (SPT) of $G$ rooted at a certain
distinguished
source vertex $r$. Moreover, we denote by $\subtree{G}{r}{x}$ the subtree of
$\tree{G}{r}$ rooted at vertex $x$.
Whenever the graph $G$ and/or the source vertex $r$ are clear from the context,
we might omit them, i.e., we write $\pi(u)$ and $d(u)$ instead of $\pi_G(r,u)$
and $d_G(r,u)$, respectively. When considering an edge $(x,y)$ of an SPT, we
assume $x$ and $y$ to be the closest and the furthest endpoints from $r$,
respectively.
Furthermore, if $P$ is a path from $x$ to $y$ and $Q$ is a path from $y$ to $z$,
with $x,y,z \in V(G)$, we denote by $P \circ Q$ the path from $x$ to $z$
obtained by concatenating $P$ and $Q$. We also denote by $w(P)$ the total
weight
of the edges in $P$.

For the sake of simplicity we consider only edge weights that are strictly
positive. However, our entire analysis also extends to non-negative weights.
Throughout the rest of the paper, we assume that, when multiple shortest
paths exist, ties are broken in a consistent manner. In particular we fix an SPT
$T=\tree{G}{r}$ of $G$ and, given a graph $H \subseteq G$ and $x,y \in V(H)$,
whenever we compute the path $\pi_H(x,y)$ and ties arise, we prefer edges in
$E(T)$.

A path between any two vertices $u,v \in V(G)$ is said to be an
$\alpha$--approximate shortest path if its length is at most $\alpha$
times the length of the shortest path between $u$ and $v$ in $G$.
For the sake of simplicity, we assume that, if a set of at most $f$ edge failures has
to be handled, the original graph is $(f+1)$--edge connected.
Indeed, if this is not the case, we can guarantee the $(f+1)$--edge connectivity
by adding at most $O(nf)$ edges of weight $+\infty$ to $G$.
Notice that this is not actually needed by any of the proposed algorithms.

\section{Our \texttt{PASPT} Structure and the Corresponding Oracle}
\label{sec:f_structure}
In what follows, we give a high-level description of our algorithm for
computing a \pspt{(2|F|+1)}, namely $H$ (see Algorithm~\ref{alg:f_path}), where
$|F| \le f$.
We define the level $\ell(v)$ of a vertex $v \in V(G)$ to be the hop-distance
between $r$ and $v$ in $T = \tree{G}{r}$, i.e., the number of edges of the
unique
path from $r$ to $v$ in $T$.
Note that, when a failure of $|F|$ consecutive edges occurs on a shortest
path, $T$ will be broken into a forest $\C$ of $|F|+1$ subtrees. We
consider these subtrees as rooted according to $T$, i.e., each tree
$T_i$ is rooted at vertex $r_i$ that minimizes $\ell(r_i)$.

Roughly speaking, the algorithm considers all possible path failures $F^*$ of
$f$ vertices by fixing the deepest endpoint $v$ of the failing path. It then
reconnects the resulting $f + 1$ subtrees of $G-F^*$ by selecting at most
$O(f^2)$ edges into a graph $U$, one for each couple of trees $T^*_i, T^*_j$
of the forest $G-F$. These edges are either directly added to the structure $H$
or they are first sparsified into a graph $U^\prime$ by using a suitable
multiplicative $(2k-1)$-spanner, so that only $k f^{1+\frac{1}{k}}$ of them are 
added to
$H$.

In particular, it is known that, given an $n$-vertex graph and an integer $k \ge 1$, both a $(2k-1)$--spanner and a $(2k-1)$--approximate distance oracle of size $O(kn^{1+\frac{1}{k}})$ can be built in $O(n^2)$ time. The oracle can report an approximate distance between two vertices in $O(k)$ time, and the corresponding approximate shortest path in time proportional to the number of its edges. For further details we refer the reader to \cite{DBLP:conf/soda/BaswanaS04,DBLP:journals/talg/BaswanaS06,DBLP:conf/icalp/RodittyTZ05}.
Recently, it has been shown in \cite{DBLP:conf/stoc/Chechik14}
that a randomized $(2k-1)$--approximate distance oracle of \emph{expected} size
$O(kn^{1+\frac{1}{k}})$ can be built, so that answering a distance query
requires only constant time. In what follows, however, we only describe results
which are based on deterministic construction and provide a worst case guarantee
on the size of the resulting structures.

\begin{algorithm}[t]
\DontPrintSemicolon
\SetKwInOut{Input}{Input}
\SetKwInOut{Output}{Output}
\Input{A graph $G$, $r \in V(G)$, an SPT $T = \tree{G}{r}$, an integer $f$}
\Output{A \pspt{(2|F|+1)} of $G$ rooted at $r$}
\BlankLine
$H \gets T = \tree{G}{r}$\;
\ForEach{$v \in V(G)$ \label{ln:main_loop}}{
	Let $\langle r = z_0, z_1, \dots, z_\ell(v) \rangle$ be the path from
$r$ to $v$ in $T$\;
	\tcp{$F^*$ contains last $\min\{f, \ell(v) \}$ edges of the path}
	Let $F^* = \{(z_{i-1}, z_i) : i > \ell(v) - \min\{\ell(v), f\}\}$ \;
	Let $\C^* = \{T^*_1, T^*_2, \dots\}$ be the set of connected components
of $T-F^*$\;
	\BlankLine
	\tcp{Build an auxiliary graph $U$ associated with $v$}
	$U \gets ( \{r^*_i \, : r^*_i \mbox{ is the root of }  T^*_i \}
,\emptyset)$ \;
  	\ForEach{$T^*_i,T^*_j \in \C^* \, : \, T^*_i \neq T^*_j$}{
	Let $E_{i,j} = \{ (u,v) \in E(G) \setminus F^* : u \in V(T^*_i), v \in
V(T^*_j)\}$ \;
	\BlankLine
	$(x',y') \gets \underset{(x,y) \in E_{i,j}}{\arg\min} \{d_T(r^*_i,
x) + w_G(x,y) + d_T(y,r^*_j)\}$ \label{ln:formula}\;
	\tcp{We say that $(x', y') \in E(G)$ is associated to $(r^*_i,
r^*_j) \in E(U)$}
    	$E(U) \gets E(U) \cup \{ (r^*_i, r^*_j) \}$\;
    	$w_U(r^*_i,r^*_j) =  d_T(r^*_i, x') + w_G(x',y') + d_T(x',r^*_j)$\;
    }
    \BlankLine
    \tcp{Optional step, executed only if $k \neq 1$. Otherwise, let $U^\prime = U$.}
	$U^\prime \gets $ Compute a $(2k-1)$-spanner of $U$ \label{ln:sparsify}\;
   	$E(H) \gets E(H) \cup E(U^\prime)$ \;
}
\Return $H$
\caption{Algorithm for building a \pspt{(2|F|+1)}. Notice that an optional integer parameter $k \ge 1$ is used. By default we set $k=1$.}
\label{alg:f_path}
\end{algorithm}

We start by bounding the running time of  Algorithm~\ref{alg:f_path}:
\begin{lemma}
	Algorithm~\ref{alg:f_path} requires $O(n(m+f^2))$ time.
\end{lemma}
\begin{proof}
	Notice that the loop in line~\ref{ln:main_loop} considers each vertex of $G$ at most once.
	We bound the time required by each iteration.
	For each vertex $v$ a complete auxiliary graph $U$ of $O(f)$ vertices is built.
	Moreover, the weights of all the edges of $U$ can be computed in $O(m)$ time by scanning
	all the edges of $E(G) \setminus F^*$ while keeping track, for each pair of vertices $r^*_i, r^*_j \in V(U)$,
	of the minimum value of the formula in line~\ref{ln:formula}.
	Finally, the optional spanner construction invoked by line~\ref{ln:sparsify} requires $O(f^2)$ time.
	This concludes the proof.
\end{proof}

We now bound the size of the returned structure:
\begin{lemma}
\label{lemma:f_path:size}
The structure $H$ returned by Algorithm~\ref{alg:f_path} contains $O(k n \cdot
f^{1+\frac{1}{k}})$ edges.
\end{lemma}
\begin{proof}
At the beginning of the algorithm, $H$ coincides with $T = \tree{G}{r}$, so
$|E(H)|=O(n)$. Therefore, we only need to bound the number of edges added to
$H$ during the
execution of the algorithm. Notice that, for each vertex $v \in V(G)$,
Algorithm~\ref{alg:f_path} considers at most $f+1$ connected components of
$\mathcal{C^*}$. For each pair of components, at most one edge is added to $U$,
hence $|E(U)|=O(f^2)$.
Either $k=1$ and $U^\prime = U$ or $k > 1$ and $U^\prime$ is a $(2k-1)$--spanner of $U$.
In both cases we have $|U^\prime| = O(k |U|^{1+\frac{1}{k}}) = O(k f^{1+\frac{1}{k}})$.
As only the edges of $U^\prime$ gets added to $H$, the claim
follows.
\end{proof}

We now upper-bound the distortion provided by the structure $H$.
For the sake of clarity, we first discuss the case where the step of
line~\ref{ln:sparsify} of Algorithm~\ref{alg:f_path} is omitted, i.e., we
simply set $k=1$ and $U^\prime = U$. At the end of this section we will argue about the general
case.

For each path failure $F$ of $|F| \leq f$ edges, and for each target
vertex $t$, we will consider a suitable path $P$ in $G-F$, whose length is at
most $(2|F|+1)$ times the distance $d_{G-F}(t)$.
Then, since $P$ might not be entirely contained in $H-F$, we will show that its
length
must be an upper bound to the length a path $Q$ in $H-F$  between $r$ an $t$,
and hence to $d_{H-F}(t)$.

We first discuss how $P$ is defined: consider the forest $\C$ of the connected
components of $T-F$.
Let $\pi = \pi_{G-F}(r)$, let $r_0 = r$, and let $t_0$ be the last vertex of
$\pi$ belonging to $T_0$. W.l.o.g., we assume $t \not\in V(T_0)$, as otherwise
we have $d_{H-F}(t) = d_{G-F}(t)$. Moreover, we call $(t_0, s_1)$
the edge following vertex $t_0$ in $\pi$.

Initially, we set $P_0 = \pi_T(s,t_0) \circ (t_0, s_1)$ and $i=1$. We proceed
iteratively: Let $T_i$ be the subtree of $C$ which contains $s_i$ and let $t_i$
be the last vertex of $\pi$ such that $t_i$ belongs to $T_i$, i.e., $t_i$ is in
the same subtree as $s_i$ (notice that, it may be that $s_i=t_i)$. Call
$r_i$ the root of $T_i$. If $t_i = t$ we set $P = P_{i-1} \circ \pi_T(s_i, r_i)
\circ \pi_T(r_i, t_i)$, and we are done. Otherwise, let $(t_i, s_{i+1})$ be the
edge following $t_i$ in $\pi$. We set $P_i = P_{i-1} \circ \pi_T(s_i, r_i)
\circ
\pi_T(r_i, t_i) \circ (t_i, s_{i+1})$, we increment $i$ by one, and we repeat
the whole procedure. Figure~\ref{fig:aux_path} shows an example of such a path
$P$. Let $h$ be the final value of $i$, at the end of this procedure, so that
$t=t_h \in V(T_h)$.
\begin{figure}
\centering
\scalebox{0.8}{\includegraphics{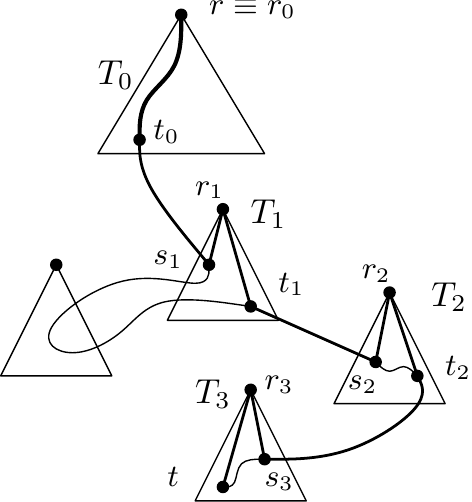}}
\caption{Example of construction of $P$.
The path $P$ is shown in bold, while the path $\pi$ is composed of both the
light subpaths and of the bold edges with endpoint in different subtrees. In this
example $P$ traverses $4$ subtrees and hence $h=3$.}
\label{fig:aux_path}
\end{figure}
Notice that, by construction, the path $P$ does not contain any failed edge.
We now argue that the length $w(P)$ of $P$, is always at most
$(2|F|+1)$ times the distance $d_{G-F}(t)$.

\begin{lemma}
\label{lemma:f_path:stretch_1}
$d_P(t) \leq (2|F|+1) \cdot d_{G-F}(t)$, for every $t \in V(G)$.
\end{lemma}

\begin{proof}
We proceed by showing, by induction on $i$, that $d_P(t_i) \leq
(2i+1) \cdot d_{G-F}(t_i)$. The claim follows since $t=t_h$ and $h \le |F|$.

The base case is trivially true, as we have $d_P(t_0) = 1 \cdot d_{G-F}(t_0)$,
since $t_0$ belongs to the same subtree $T_0$ as $r$. Now, suppose that the
claim is true for $i-1$. We can prove that it is true also for $i$ by writing:
\begin{align*}
  d_P(t_i) & = d_P(t_{i-1}) + d_P(t_{i-1}, s_i) + d_P(s_i, r_i) + d_P(r_i, t_i)
\\
  & \le (2i-1) \cdot d_{G-F}(t_{i-1}) + d_{G-F}(t_{i-1}, s_i) + d_G(s_i, r_i) +
d_G(r_i, t_i) \\
  & \le (2i-1) \cdot d_{G-F}(t_{i-1}) + d_{G-F}(t_{i-1}, s_i) + d_G(s_i, t_i) +
2 d_G(r_i, t_i) \\
  & \le (2i-1) \cdot d_{G-F}(t_i) +  2 d_{G}(t_i)) \le (2i+1) \cdot
d_{G-F}(t_i).
\end{align*}
\end{proof}

It remains to show that, even though $P$ might not be entirely contained in
$H-F$,
its length $w(P)$ is always an upper bound to $d_{H-F}(t)$.

Let $v$ be the deepest endpoint (w.r.t. level) among the endpoints of the edges
in $F$. Moreover, let $F^*$ be the set of failed edges considered by
Algorithm~\ref{alg:f_path} when $v$ is examined at line~\ref{ln:main_loop}, and
let $U$ be the the corresponding auxiliary graph. Notice that $F \subseteq F^*$
as $F^*$ always contains $\min\{ \ell(v), f \}$ edges. As a consequence, $T_0
\in \C$ contains, in general, several trees in $\C^*$. We let $R$ be the set of
the roots of all the subtrees of $T_0$ which are in $\C^*_0$. Notice that every
other tree $T_j \in C$ such that $T_j \neq T_0$ belongs to $\C^*$ (see
Figure~\ref{fig:path_to_ri}).

Remember that $r_h$ is the root of the subtree $T_h \in \C^* = T-F^*$ which
contains $t$. Let $r^\prime_0$ be the root of the last tree $T^\prime_0 \in
\C^*$ which is contained in $T_0$ and is traversed by $\pi_{G-F}(r_h)$. It
follows that $r^\prime_0 \in V(P)$. We now construct another path $Q$, which
will be entirely contained in $H-F$. We choose a special vertex $r^*_0 \in R$,
as follows:
\begin{equation}
	\label{eq:choice_of_r^*_0}
	r^*_0 = \arg \min_{z \in R} \{ d_T(z) + d_U(z, r_h) \}.
\end{equation}

\begin{figure}
	\centering
	\scalebox{0.72}{\includegraphics{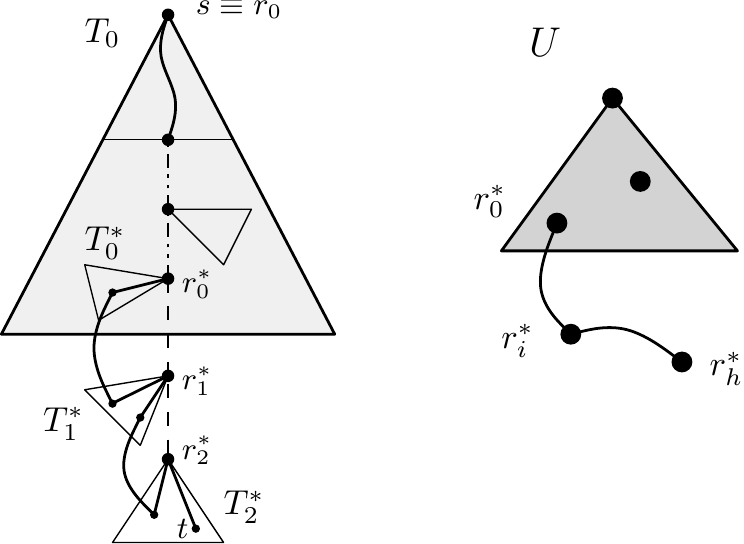}}
	\caption{An example of path $Q$ contained in $H-F$ (left) and of the corresponding
    edges of $U$ (right). The length  of $Q$ is upper-bounded by that of $P$.}
	\label{fig:path_to_ri}
\end{figure}

The path $Q$ is composed of three parts, i.e. $Q = Q_1 \circ Q_2 \circ Q_3$.
The
first one, $Q_1$, coincides with $\pi_T(r^*_0)$. The second one is obtained by
considering the shortest path $\pi_U(r^*_0, r_h)$ and by replacing each edge
going from a vertex $r^*_i \in V(U)$ to a vertex $r^*_j \in V(U)$ with the
path:
$\pi_T(r^*_i, x') \circ (x',y') \circ \pi_T(x', r^*_j)$, where $(x',y')$ is the
edge associated to $(r^*_i, r^*_j)$ by Algorithm~\ref{alg:f_path} when $v$ is
considered. Finally, $Q_3 = \pi_T(r^*_h, t)$. In Figure~\ref{fig:path_to_ri},
we show an example of how such path $Q$ can be obtained.
We now prove that:
\begin{lemma}
\label{lemma:f_path:stretch_2}
$d_{H-F}(r,t) \le w(Q) \le w(P)$
\end{lemma}
\begin{proof}
Notice that the path $Q$ is in $H$ and does not contain any failed edge, hence
$d_{H-F}(r,t) \le w(Q)$ is trivially true.

To prove $w(Q) \le w(P)$, notice that $P$ can also be decomposed into the three
subpaths $P_1 = P[r, r^\prime_0]$, $P_2 = P[r^\prime_0, r_h]$ and $P_3 =
P[r_h,t]$. We have that that $P_3 = Q_3$ and that the endpoints of $P_2$ coincide
with the endpoints of $Q_2$.
By the choice of $r^*_0$, we must have $w(Q_1) + w(Q_2) \le w(P_1) + w(P_2)$ as
the (weighted length of) path $P_1 \circ P_2$ is considered in equation
\eqref{eq:choice_of_r^*_0} when $z=r^\prime_0$.
This implies that $w(Q) = w(Q_1) + w(Q_2) + w(Q_3) \le w(P_1) + w(P_2) + w(P_3) = w(P)$.
\end{proof}

By combining~Lemma~\ref{lemma:f_path:size} with
Lemma~\ref{lemma:f_path:stretch_1} and \ref{lemma:f_path:stretch_2}, it
immediately follows:
\begin{theorem}
Algorithm \ref{alg:f_path} computes, in $O(n(m+f^2))$ time, a
\pspt{(2|F|+1)} of size $O(n f^2)$, for any $|F|\le f$.
\end{theorem}

We now relax the assumption that $U = U^\prime$. Indeed, if $k \neq 1$,
Algorithm~\ref{alg:f_path} computes, in line~\ref{ln:sparsify}, a
$(2k-1)$--spanner $U'$ of the graph $U$.
In this case, we can construct a path $Q^\prime$ in a similar way as we did for $Q$,
with the exception that we now use the graph $U^\prime$ instead of $U$. Once we do so,
it is easy to prove that a more general version of
Lemma~\ref{lemma:f_path:stretch_2} holds:

\begin{lemma}\label{lemma:f_path:stretch_2-bis}
  $d_{H-F}(r,t) \le  (2k-1) w(Q^\prime) \le  (2k-1)  w(P)$
\end{lemma}

Lemma \ref{lemma:f_path:stretch_2-bis}, combined with Lemma~\ref{lemma:f_path:stretch_1},
immediately implies that $d_{H-F}(r,t) \le (2k-1)(2|F|+1) d_{G-F}(r,t)$. This discussion
allows us to show an interesting trade-off between the size of the returned
structure and the multiplicative stretch provided, as summarized by the
following theorem:

\begin{theorem}
Let $k\ge 1$ be an integer. Then, Algorithm \ref{alg:f_path} can compute, in
$O(n(m+f^2))$ time, a \pspt{(2k-1)(2|F|+1)} of size $O(n k\cdot
f^{1+\frac{1}{k}})$.
\end{theorem}

\subsection{Oracle Setting}
In what follows, we show how Algorithm \ref{alg:f_path} can be used to compute an approximate 
distance oracle of size $O(n f^2)$ (see Algorithm \ref{alg:f_path_oracle}).
We also show that a smaller-size oracle can be obtained (see Algorithm 
\ref{alg:f_path_oracle_2}) 
if we allow for a slightly larger query time.

\begin{algorithm}[t]
\DontPrintSemicolon
\SetKwInOut{Input}{Input}
\SetKwInOut{Output}{Output}

Preprocess  $T = \tree{G}{r}$ to answer LCA queries as shown in
\cite{HT84}\;
For each vertex $v \in V(G)$, compute and store its level $\ell(v)$.
\BlankLine	

\ForEach{$v \in V(G)$}{
  Let $\langle r = z_0, z_1, \dots, z_\ell(v)$ be the path from $r$ to
  $v$ in $T$\;
  \BlankLine
  Build graph $U$ associated with vertex $v$ as in Algorithm~\ref{alg:f_path} \;
  Compute and store the solution to the all-pairs shortest paths problem on $U$\;
  \BlankLine
  \ForEach{$\eta = 1, \dots, \min\{f, \ell(v)\}$}{
      \ForEach{ $r_h : h > \ell(v)-\eta$ }{
	$R \gets \{ z_i : 0 \le i \le \ell(v)-\eta \}$ \;
	Let $r^*_0$ be the vertex of $R$ minimizing
	Equation~\eqref{eq:choice_of_r^*_0} \;
	Store $r^*_0$ with key $(v, \eta, r_i)$ \label{ln:store_r*_0}\;
      }
  }
}

\caption{Algorithm for building an oracle with constant query time.}
\label{alg:f_path_oracle}
\end{algorithm}

\begin{algorithm}[t]
\DontPrintSemicolon
\SetKwInOut{Input}{Input}
\SetKwInOut{Output}{Output}

Preprocess  $T$ to answer LCA queries as shown in \cite{HT84}\;
For each vertex $v \in V(G)$, compute and store its level $\ell(v)$.
\BlankLine	

\ForEach{$v \in V(G)$}{
Build graph $U$ associated with vertex $v$ as in Algorithm~\ref{alg:f_path} \;
Build and store a distance sensitivity oracle of $U$ with stretch $2k-1$ \;
}
\caption{Algorithm for building an oracle with $O(f)$ query time.}
\label{alg:f_path_oracle_2}
\end{algorithm}

\begin{theorem}
\label{thm:f_path_oracle}
Let $F$ be a path failure of $|F| \le f$ edges and $t \in V(G)$.
Algorithm~\ref{alg:f_path_oracle} builds, in $O(n(m+f^2))$ time, an oracle of size $O(n f^2)$
which is able to return:
\begin{itemize}
\item  a $(2|F|+1)$-approximate distance in $G-F$ between $r$ and $t$ in constant time;
\item  the associated path in a time proportional to the number of its edges.
  \end{itemize}
\end{theorem}
\begin{proof}
In order to answer a query we need to find: (i) the root $r^*_0$ of the
subtree of $\C^*$ which contains $t_0$, (ii) the root $r_h$ of the subtree of
$\C^*$ containing $t$.
In order to find $r_h$, we perform a LCA query on $T$ to find the least common
ancestor $u$ between $v$ and $t$. Either  $\ell(v) \ge \ell(u) > \ell(v) -
|F|$,
in which case $u = r_h$, or $\ell(u)  \le \ell(v)-|F|$ which means that $t$
belongs to $T_0$.
As in the latter case we can simply return $d_T(t)$, we focus on the
former one.
To find $r^*_0$ we look for the vertex associated with the triple $(v,
|F|, r_h)$ stored by Algorithm~\ref{alg:f_path_oracle} at line
\ref{ln:store_r*_0}.

We answer a distance query with the quantity $d_T(r^*_0) +
d_{U^\prime}(r^*_0, r^*_h) + d_T(r_h, t)$, which can be computed in constant
time by accessing the distances stored in shortest path tree $T$, plus the
solution of the APSP problem on $U^\prime$ computed by
Algorithm~\ref{alg:f_path_oracle} when vertex $v$ was considered.

To answer a path query we simply construct, and return, the path $Q$, by
expanding the edges of the graph $U^\prime$ into paths which are in $G-F$, as
explained before. This clearly takes a time proportional to the number of edges
of $Q$.	
\end{proof}

If we allow for a query time that is proportional to $O(f+k)$, we can reduce
the size of the oracle by computing a distance sensitivity oracle (DSO) of $U$
(see
Algorithm~\ref{alg:f_path_oracle_2}). In this case, we can still find vertex
$r_h$ using the LCA query, as shown in the proof of
Theorem~\ref{thm:f_path_oracle}, while vertex $r^*_0$ is guessed among the (up
to) $f$ roots of the trees in $G-F^*$ which are contained in $T_0$. The
resulting oracle is summarized by the following:

\begin{theorem}
\label{thm:f_path_spanned_oracle}
Let $F$ be a path failure of $|F| \le f$ edges, let $t \in V(G)$ and let $k \ge
1$ be an integer.
Algorithm~\ref{alg:f_path_oracle_2} builds, in $O(n(m+f^2))$ time, an oracle of size $O(n k f^{1+\frac{1}{k}})$ which is able to return:
	\begin{itemize}
\item a $(2k-1)(2|F|+1)$-approximate distance in $G-F$ between $r$ and $t$ in $O(f+k)$ time;
\item the corresponding path in $O(\ell + k + f)$ time, where $\ell$ is the 
number of its edges.
	\end{itemize}
\end{theorem}

\section{Our \pspt{3} Structure for Paths of 2 Edges}
\label{sec:2_structure}
\newcommand{\FirstLast}{\ensuremath{\texttt{FirstLast}}}

In what follows, we provide an algorithm which builds a \pspt{3} (see
Algorithm~\ref{alg:2_path}) for the special case of at most $f = 2$ cascading
edge failures.
This structure improves,  w.r.t. the quality of the stretch, over the general
\pspt{(2|F|+1)} of Section~\ref{sec:f_structure}.

\begin{algorithm}[t]
\DontPrintSemicolon
\SetKwInOut{Input}{Input}
\SetKwInOut{Output}{Output}
\Input{A graph $G$, $r \in V(G)$, an SPT $T = \tree{G}{r}$}
\Output{A \pspt{3} of $G$ rooted at $r$}
\BlankLine	
$H \gets \tree{G}{r}$ \;
$\hat{T}  \gets$ compute a $3$-EASPT of $\tree{G}{r}$ as shown in \cite{NPW03}
\;
$H \gets E(H) \cup E(\hat{T})$  \;
Compute a path decomposition $\mathcal{P}$ of $\tree{G}{r}$ by recursively
applying
Lemma~\ref{lemma:path_decomposition_one_path}\;
\BlankLine
\ForEach{Path $P \in \mathcal{P}$}{
  \ForEach{$x \in V(P) \, : x$ is not a leaf and $x \neq r$}{
  Let $z$ be the (unique) child of $x$ in $P$\;
  Let $\hat{e}$ be the edge connecting $x$ and its parent int $T$\;
  \BlankLine

    \tcp{Protect vertex $x$}
    $E(H) \gets E(H) \cup \FirstLast(\pi_{G-\hat{e}}(x), \subtree{G}{r}{z})$
\label{ln:protecting_x_start}\;
   	\If{$\pi_{G-\hat{e}}(x)$ contains an edge $e^\prime$ in $C(x)$}
   	{
    	$E(H) \gets E(H) \cup \FirstLast(\pi_{G-\hat{e}-e^\prime}(x),
\subtree{G}{r}{z})$ \label{ln:protecting_x_end}\;
    }
    \BlankLine
    \tcp{Protect vertex $z$}
    $E(H) \gets E(H) \cup E(\pi_{G-\hat{e}}(z))$ \label{ln:protecting_z_start}\;
    \ForEach{$e^\prime \in  \{\pi_{G-\hat{e}}(z) \cap C(x)\}$}{
$E(H) \gets E(H) \cup E(\pi_{G - \hat{e}- e^\prime }(z))$
\label{ln:protecting_z_end} \;
	}

    \BlankLine
    \tcp{Protect all the other children of $x$}
    \ForEach{children $z_i$ of $x \, \:\, z_i \neq z$
\label{ln:protecting_children_start}}
    {
    	Let $(u,q)$ be the first edge of $\pi_{G-\hat{e}-(x,z_i)}(x,z_i)$ with
$q \in
V(\subtree{G}{r}{z_i})$\;
    	$E(H) \gets E(H) \cup \{ (u,q) \}$ \label{ln:protecting_children_end} \;
	}
	
    \BlankLine
    \tcp{Protect vertices whose paths that do not contain $x$}
    $T^\prime \gets \subtree{G-x}{r}$ with edges oriented towards the leaves
\label{ln:protecting_other_start} \;
    $E(H) \gets E(H) \cup \{ (x_1, x_2) \in E(T^\prime) \, : \, x_2 \not\in
\subtree{G}{r}{z} \}$ \label{ln:protecting_other_end} \;
	}
}

\Return $H$
\caption{Algorithm for building a \pspt{3} for the case of $f=2$.}
\label{alg:2_path}
\end{algorithm}

The algorithm starts with a 3-\texttt{EASPT} with $O(n)$ edges \cite{NPW03} and
proceeds as follows. As initial building block, it considers a suitable path
$P$ in the shortest-path tree $\tree{G}{r}$, and constructs a structure $H$ that
is able to handle the failure of a pair of edges $\{e_1,e_2\}$, such that $e_1
\in P$, and guarantees $3$-stretched distances from $r$, for each vertex in $G$.
Then, we make use of the following result of \cite{BK13}:
\begin{lemma}[\cite{BK13}]
\label{lemma:path_decomposition_one_path}
There exists an $O(n)$ time algorithm to compute an ancestor-leaf path $Q$ in
$\tree{G}{r}$ whose removal splits $\tree{G}{r}$ into a set of disjoint subtrees
$\subtree{G}{r}{r_1},\dots,\subtree{G}{r}{r_j}$ such that, for each $i\le j$:
\begin{itemize}
\item $| \subtree{G}{r}{r_i}| < n/2$ and $V(Q) \cap V(\subtree{G}{r}{r_i}) =
\emptyset$
\item $\subtree{G}{r}{r_i}$ is connected to $Q$ through some edge for each $i\le
j$
\end{itemize}
\label{lemma:path_dec}
\end{lemma}

This allows us to incrementally add edges to $H$ by considering a set
$\mathcal{P}$ of edge-disjoint paths. This set can be obtained by
recursively using the path decomposition technique of
Lemma~\ref{lemma:path_dec} on the shortest-path tree $\tree{G}{r}$. We show
that, in this way, we are able to build a \pspt{3} of size $O(n \log n)$.
Given a path $\pi = \langle s, \dots, t \rangle$ and a tree $T^\prime$, we
denote by $\FirstLast(\pi, T^\prime)$ the edges of the subpaths of $\pi$ going
(i) from $s$ to the first vertex of $\pi$ in $V(T^\prime)$, and (ii) from the
last vertex of $\pi$ in $V(T^\prime)$ to $t$. If these vertices do not exists,
i.e., $V(\pi) \cap V(T^\prime) = \emptyset$, then we define $\FirstLast(\pi,
T^\prime) = E(\pi)$. Moreover, we denote by $C(x)$ the edges connecting vertex
$x$ to its children in \tree{G}{r}. We are able to prove the following theorem,
whose proof is given in the appendix:

\begin{theorem}
\label{thm:2_path_preproc}
Let $F$ be a path failure of $|F| \le 2$ edges and $t \in V(G)$. 
Algorithm \ref{alg:2_path} computes, in $O(nm + n^2 \log n)$ time, a \pspt{3}
of size $O(n \log n)$.
\end{theorem}
Notice that it is possible to modify Algorithm~\ref{alg:2_path}
in order to build an oracle of size $O(n \log n)$ which is able to report,
with optimal query time, both a $3$-stretched shortest path in $G-F$ and its
distance, when $F$ contains two consecutive edges in $T$. Both the description
of the modified algorithm and the proof of the following theorem is given
in the appendix.
\begin{theorem}
\label{thm:oracle_2_path}
Let $F$ be a path failure of $|F| \le 2$ edges and $t \in V(G)$. A modification
of Algorithm~\ref{alg:2_path} builds, in $O(nm + n^2 \log n)$  time, an oracle
of size $O(n \log n)$ which is able to return:
\begin{itemize}
\item a $3$-approximate distance in $G-F$ between $r$ and $t$ in constant time;
\item the associated path in a time proportional to the number its edges.
\end{itemize}
\end{theorem}

\section{Experimental Study}\label{sec:experiments}
In this section, we present an experimental study to assess the performance,
w.r.t. both the quality of the stretch and the size (in terms of edges), of
the proposed structures within SageMath (v. 6.6) under GNU/Linux.

As input to our algorithms, we used weighted undirected graphs belonging to the
following graph categories: (i) \emph{Uncorrelated Random Graphs} (ERD):
generated by the general \emph{\erdos} algorithm~\cite{B01}; (ii)
\emph{Power-law Random Graphs} (BAR): generated by the \emph{\barabasi}
algorithm~\cite{AB99}; \emph{Quandrangular Grid Graphs} (GRI): graphs whose
topology is induced by a two-dimensional grid formed by squares.
For each of the above synthetic graph categories we generated three input
graphs of different size and density. We assigned weights to the edges at
random, with uniform probability, within $[100,100\,000]$.
We also considered two real-world graphs. In
details: (i) a graph (CAI) obtained by parsing the \emph{CAIDA
IPv4 topology dataset}~\cite{caida}, which describes a subset of the Internet
topology at router level (weights are given by round trip times);
(ii) the road graph of Rome (ROM) taken from the 9th Dimacs Challenge
Dataset\footnote{http://www.dis.uniroma1.it/challenge9} (weights are given by
travel times).

\renewcommand{\arraystretch}{1.1}

\newcolumntype{L}[1]{>{\raggedright\let\newline\\\arraybackslash\hspace{0pt}}m{
#1}}
\newcolumntype{C}[1]{>{\centering\let\newline\\\arraybackslash\hspace{0pt}}m{#1}
}
\newcolumntype{R}[1]{>{\raggedleft\let\newline\\\arraybackslash\hspace{0pt}}m{#1
}}

\newcolumntype{d}[1]{D{.}{\cdot}{#1} }

\newcommand{\ccol}[1]{\multicolumn{1}{c|}{#1}}
\newcommand{\ccolnop}[1]{\multicolumn{1}{c}{#1}}
\newcommand{\szccolp}[2]{\multicolumn{1}{C{#1cm}|}{#2}}
\newcommand{\szccolnop}[2]{\multicolumn{1}{C{#1cm}}{#2}}

Then, for each input graph, we built both the \pspt{(2k -1)(2|F|+1)}, for which 
we focused on the basic case of $k=1$, and the \pspt{3}, as follows: we 
randomly chose
a root vertex, computed the SPT and enriched it by using the
corresponding procedures (i.e. Algorithm~\ref{alg:f_path} and~\ref{alg:2_path},
resp.). We measured the total number of edges of the resulting structures.

Regarding Algorithm~\ref{alg:f_path}, we set $f=10$, as such a
value has already been considered in previous works focused on the effect
of path-like disruptions on shortest paths~\cite{BW09,DDFLP14}.
Then, we randomly select path failures of $|F|$ edges to perform on the input
graphs, with $|F|$ uniformly chosen at random within the range $[2,f]$. We
removed the edges belonging to the path failure from both the original graph
and the computed structure.
Regarding Algorithm~\ref{alg:2_path}, we simply chose at random a pair of edges
and removed them from both the original graph and the computed structure.

After the removal, we computed distances, from the root vertex, in both the
original graph and the fault tolerant structure, and measured the resulting
average stretch. In order to be fair, we considered only those nodes that get
disconnected as a consequence of the failures. Our results are summarized in
Table~\ref{table:results}, where, for each input graph, we report the number of
vertices and edges, the average size
(number of edges) of the two fault tolerant structures and the corresponding
provided average stretch.

\begin{table}[t]
\centering
\begin{tabular}{c|c|c|c|c|c|c}
\multirow{2}{*}{\bf G} &
\multirow{2}{*}{\bf |V(G)|} &
\multirow{2}{*}{\bf |E(G)|} &
\multicolumn{2}{c|}{\pspt{(2|F|+1)}}&
\multicolumn{2}{c}{\pspt{3}}\\
\cline{4-7}
& & & \bf \#edges & \bf avg stretch& \bf \#edges & \bf avg stretch\\
\hline
ERD-1 & 500 & 50\,000 & 3\,980 &1.8015 & 957 & 1.0000\\
\hline
ERD-2 & 1\,000 &50\,000 & 8\,899 &1.1360 & 1\,924 & 1.0000\\
\hline
ERD-3 & 5\,000 & 50\,000 & 20\,198 &1.0903 &  9\,501 & 1.0035 \\
\hline
BAR-1 & 500 & 1\,491 & 1\,366 &1.0003 & 949 & 1.0041\\
\hline
BAR-2 & 1\,000 & 2\,991 & 2\,765&1.0034 & 1\,871 & 1.0005 \\
\hline
BAR-3 & 5\,000 &14\,991 & 13\,349&1.0040 & 9\,459 & 1.0000\\
\hline
GRI-1 & 500 & 1\,012 & 1\,008 & 1.0005 & 868& 1.0000\\
\hline
GRI-2 & 1\,000 &1\,984 & 1\,973 & 1.0000 &1\,749 & 1.0000\\
\hline
GRI-3 & 5\,000 &9\,940 & 9\,884& 1.0000 &8\,826 & 1.0000\\
\hline
CAI &5\,000 & 6\,328 & 6\,033 & 1.0000 & 6\,026 & 1.0000\\
\hline
ROM & 3\,353 & 4\,831 & 4\,796 & 1.0000 & 4\,780 & 1.0000 \\
\end{tabular}
\caption{Average number of edges and stretch factor for both the
\pspt{(2|F|+1)} and the \pspt{3}.
}
\label{table:results}
\end{table}

First of all, our results show that the quality of the stretch, provided by
both the \pspt{(2|F|+1)} and the \pspt{3} in practice, is always by far better
than the estimation given by the worst-case bound (i.e. 2|F|+1 and 3,
resp.). In details, the average stretch is always very close to $1$ and
does not depend neither on the input size nor on the number of failures. This
is probably due to the fact that those cases considered in the worst-case
analysis are quite rare.

Similar considerations can be done w.r.t. the number of edges that are added to
the SPT by Algorithms~\ref{alg:f_path} and~\ref{alg:2_path}. In fact, also in
this case, the structures behave better than what the worst-case bound
suggests. For instance, the number of edges of the \pspt{(2|F|+1)} (the
\pspt{3}, resp.) is much smaller than $n f^2$ ($n \log n$, resp.).
In summary, our experiments suggest that the proposed fault tolerant structures
might be suitable to be used in practice.

\bibliographystyle{plain}
\bibliography{biblio}
 \appendix
 \section{Omitted Proofs}
In this section, we upper-bound the running time of 
Algorithm~\ref{alg:2_path}. In details, we prove that, given a set of two 
failures $F = \{e_1,e_2\}$, $d_{H-F}(t) \le 3\cdot d_{G-F}(t)$ for every $t\in 
V(G)$, and that $H$ contains $O(n\cdot \log n)$ edges.\footnote{We only focus 
on 
exactly two edge faults since $H$ already contains a 3-\texttt{EASPT}.} 
W.l.o.g. 
we assume that that $e_1=(y,x)$, $e_2=(x, k)$, where $x$ is a child of $y$ and 
$k$ is a child of $x$ in $T$.

Notice that, every possible edge $e_1$ of a pair of failures that can occur on 
$\tree{G}{r}$ is considered exactly once as, during the construction phase, we 
make use of the path decomposition technique of~\cite{BK13}.  Let $P \in 
\mathcal{P}$ be the path of the path decomposition $\mathcal{P}$ which contains 
$e_1$ and let $z$ be the vertex following $x$ in $P$.\footnote{Note that vertex 
$z$ always exists as the last vertex of $P$ must be a leaf in $T$, while $x$ is 
an internal vertex.} Notice that the other failed edge $e_2=(x,k)$ might or 
might not belong to the very same path $P$.

We now bound the distance $d_{H-F}(t)$ between $r$ and a generic \emph{target 
vertex} $t \in V(G)$. We assume, w.l.o.g., that $t$ belongs to 
$\subtree{G}{r}{x}$ as otherwise we trivially have $d_{H-F}(t) = d_{G-F}(t)$. 
For the sake of clarity, we divide the proof into parts, depending on the 
position of $t$ in $\tree{G}{r}-F$ and on the structure of the path 
$\pi_{G-F}(t)$.

\begin{lemma}
\label{lemma:2_path_t_in_Tz}
For every $t \in V(\subtree{G}{r}{z})$, there exists a path $\pi^*(t)$ between 
$r$ and $t$ in $H-F$ such that $w(\pi^*(t)) \le 3\cdot d_{G-F}(t)$.
\end{lemma}
\begin{proof}
The edges added to $H$ at
Lines~\ref{ln:protecting_z_start}--\ref{ln:protecting_z_end} of
Algorithm~\ref{alg:2_path} guarantee that $d_{H-F}(z)$ equals $d_{G-F}(z)$ for
every possible pair of failures. It follows that we can choose
$\pi^*(t) = \pi_{G-F}(z) \circ \pi_{G}(z,t)$, as we have:
\begin{align*}
w(\pi^*(t)) & = d_{H-F}(z) + d_{H-F}(z,t)\\
	   & \le d_{H-F}(z) + d_{G}(z,t)\lcomment{$\pi_G(z,t) =
\pi_{H-F}(z,t)$}\\
	   & \le d_{G-F}(z) + d_{G}(z,t)\lcomment{By
Lines~\ref{ln:protecting_z_start}--\ref{ln:protecting_z_end} of
Alg.~\ref{alg:2_path}}\\
	   & \le d_{G-F}(t) + d_{G-F}(t,z) + d_G(z,t)\lcomment{By triang.
ineq.}\\
	   & \le d_{G-F}(t) + 2d_{G}(t,z)\lcomment{$\pi_G(z,t) =
\pi_{G-F}(z,t)$}\\
	   & \le d_{G-F}(t) + 2d_{G}(t) \le 3d_{G-F}(t). \lcomment{$z \in 
V(\pi_G)(t)$}
\end{align*}
\end{proof}

\begin{lemma}
\label{lemma:2_path_x}
There exists a path $\pi^*(x)$ between $r$ and $x$ in $H-F$ such that
$w(\pi^*(x)) \le 3\cdot d_{G-F}(x)$ if $e_2=(x,z)$ and $w(\pi^*(x)) =
d_{G-F}(x)$ otherwise.
\end{lemma}
\begin{proof}
	If $V(\pi_{G-F}(x)) \cap V(\subtree{G}{r}{z}) = \emptyset$ we set
$\pi^*(x) = \pi_{G-F}(x)$ and we are done as $\pi^*(x)$ gets added to $H$ by
Lines~\ref{ln:protecting_x_start}--\ref{ln:protecting_x_end} of
Algorithm~\ref{alg:2_path}.

	Otherwise, if $V(\pi_{G-e_1-e_2}(x)) \cap V(\subtree{G}{r}{z}) \neq
\emptyset$,  let $q,q^\prime$ be the first and last vertex of
$\pi=\pi_{G-e_1-e_2}(x)$ that is in $V(\subtree{G}{r}{z})$, respectively. If 
$e_2
\neq (x,y)$ then it suffices to choose $\pi^*(x)=\pi$. Indeed, by construction,
$\pi$ is in $H$ since both $\pi[r,q]$ and $\pi[q,x]=\pi_G(q,z)$ are in $H$.
		
	Finally, if $e_2 = (x,y)$, then $pi^*(x) = \pi^*(q^\prime) \circ
\pi[q^\prime, x]$, where $\pi^*(q^\prime)$ is the path of
Lemma~\ref{lemma:2_path_t_in_Tz}. The path $\pi^*(x)$ is in $H$ and we can
bound its length as follows:
	\[
		w(\pi^*(x)) = w(\pi^*(q^\prime)) + d_{G-F}(q^\prime, x) \le 3
d_{G-F}(q^\prime) + d_{G-F}(q^\prime, x)  \le  3 d_{G-F}(x).
	\]
\end{proof}

\begin{lemma}
\label{lemma:2_path_t_not_in_Tz_not_x}
For every $t \not\in V(\subtree{G}{r}{z}) \cup \{x\}$ such that $x \not \in
V(\pi_{G-F}(t))$, there exists a path $\pi^*(t)$ between $r$ and $t$ in $H-F$
satisfying $w(\pi^*(t)) \le 3\cdot d_{G-F}(t)$.
\end{lemma}
\begin{proof}
First of all notice that it must hold $\pi_{G-F}(t) = \pi_{G-x}(t)$. If
$\pi_{G-x}(t)$ does not contain any vertex of $\subtree{G}{r}{z}$ we are done, 
as
we can set $\pi^*(t) = \pi_{G-x}(t)$ (by
Lines~\ref{ln:protecting_other_start}--\ref{ln:protecting_other_end} of
Algorithm~\ref{alg:2_path}). Otherwise, let us call $q$ the last vertex of
$\pi_{G-x}(t)$ that belongs to $\subtree{G}{r}{z}$.
We set $\pi^*(t) = \pi^*(q) \circ \pi_{G-x}(q,t)$, where $\pi^*(q)$ is the path
of Lemma~\ref{lemma:2_path_t_in_Tz}. We have
\begin{align*}
w(\pi^*(t)) & = w(\pi^*(q)) + d_{H-x}(q,t)\\
	& \le 3d_{G-F}(q) + d_{H-F}(q,t) \lcomment{By
Lemma~\ref{lemma:2_path_t_in_Tz}} \\
	& \le 3d_{G-F}(q) + d_{H-F}(q,t) \lcomment{By
Lines~\ref{ln:protecting_other_start}--\ref{ln:protecting_other_end} of
Alg.~\ref{alg:2_path}} \\
	& \le 3d_{G-F}(q) + 3d_{G-F}(q,t) = 3d_{G-F}(t)\lcomment{Since $q \in
V(\pi_{G-F}(t))$}
\end{align*}
\end{proof}

\begin{lemma}
\label{lemma:2_path_t_not_in_Tz_x}
For every $t \not\in V(\subtree{G}{r}{z}) \cup \{x\}$ such that $x \in
V(\pi_{G-F}(t))$,
there exists a path $\pi^*(t)$ between $r$ and $t$ in $H-F$ satisfying
$w(\pi^*(t)) \le 3\cdot d_{G-F}(t)$.
\end{lemma}
\begin{proof}
Notice that $t$ belongs to a subtree $\subtree{G}{r}{z_i}$ for exactly one child
$z_i \neq z$ of $x$ in
$\tree{G}{r}$. If $(x,z_i) \neq e_2$, we have that
$\pi_{G}(x,t)=\pi_{G-F}(x,t)=\pi_{H-F}(x,t)$
We set $\pi^*(t) = \pi^*(x) \circ \pi_{G}(x,t)$ where $\pi^*(x)$ is the path of
Lemma~\ref{lemma:2_path_x}. We have:
\[
	w(\pi^*(t)) = w(\pi^*(x)) + d_G(x,t) \le 3 d_{G-F}(x) + d_{G-F}(x,t)
\le 3 d_{G-F}(t)
\]

Otherwise, $e_2 = (x,z_i)$, which means that $t$ belongs to a subtree of
$\tree{G}{r}$ which gets disconnected form $x$ by the removal of $e_2$.

Since $e_2 \neq (x, z)$, we know that the path $\pi^*(x)$ of
Lemma~\ref{lemma:2_path_x} satisfies \linebreak $w(\pi^*(x)) = d_{G-F}(x)$. 
Moreover, the
shortest path $\pi_{G-F}(x, z_i)$ traverses at most one other subtree (other
than $\subtree{G}{r}{z_i}$) rooted at a child of $x$. This is because $H-F$
contains the shortest paths from $x$ to  every vertex in $V(\subtree{G}{r}{x})
\setminus V(\subtree{G}{r}{z_i})$. Let $(u,q)$ be the first edge of the path
$\pi_{G-F}(x,z_i)$ such that $q \in V(\subtree{G}{r}{z_i})$ and notice that this
edge belongs to $H$
(Lines~\ref{ln:protecting_children_start}--\ref{ln:protecting_children_end} of
Algorithm~\ref{alg:2_path}). By the choice of $(u,q)$ we have $\pi_{H-F}(x,q) =
\pi_{G}(x,u) \circ (u,q)$. We set $\pi^*(t) = \pi^*(x) \circ \pi_{G}(x,u) \circ
(u,q) \circ \pi_{G}(q,z_i) \circ \pi_{G}(z_i,t)$.
\begin{align*}
	w(\pi^*(t)) & =  w(\pi^*(x)) + d_{G}(x,u) + w(u,q) + d_{G}(q,z_i) +
d_{G}(z_i,t)\\
	& \le d_{G-F}(x) + d_{G-F}(x,q) + d_{G-F}(q,z_i) +  d_{G}(z_i,t) \\
	& \le d_{G-F}(x) + d_{G-F}(x, z_i) + d_{G}(z_i,t)\lcomment{Since $q \in
V(\pi_{G-F}(x,z_i))$}\\
	& \le d_{G-F}(x) + d_{G-F}(x, t) + 2d_{G}(z_i,t) \lcomment{By triang.
ineq.}\\
	& \le d_{G-F}(x) + d_{G-F}(x, t) + 2d_{G}(x,t) \lcomment{$z_i \in 
V(\pi_G)(x,t)$} \\
	& \le d_{G-F}(x) + d_{G-F}(x, t) + 2d_{G-F}(x,t)\\
	& = d_{G-F}(x) +3 d_{G-F}(x,t) \le 3 d_{G-F}(t).\lcomment{Since $x \in
V(\pi_{G-F}(t))$}
\end{align*}
\end{proof}

We now bound the size of $H$. In order to do so, it is useful to split the 
vertices
of the $T$ into components, depending on the vertex $x$ that is currently 
considered by
Algorithm~\ref{alg:2_path}. More formally, when a couple of edges $(y,x), (x,z)$ 
is
considered we can partition the vertices of $T-x$ into three distinct sets (see 
Figure~\ref{fig:2_failure_structure}):
\begin{itemize}
	\item $U_x$, which contains the vertices which are in the same subtree
as $r$ in $T-x$;
	\item $D_x$, which contains the vertices which are in the subtree of
$T$ rooted at $z$:
	\item $O_x$, which contains all the vertices which are in the subtree
rooted at some child $z_i \neq z$ of $x$ in $T$.
\end{itemize}
We are now read to prove:

\begin{lemma}
\label{lemma:2_path:size}
The structure $H$ returned by Algorithm \ref{alg:2_path} contains $O(n\cdot
\log n)$ edges.
\end{lemma}
\begin{proof}
To prove the claim we fix a generic path $P = \langle u, \dots, v \rangle$ (of
at least two edges) of the path decomposition, where $v$ is a left and $u$ one
its ancestors in $T$. We show that, when Algorithm~\ref{alg:2_path}
considers $P$, the total number of edges added to $H$ is
$O(|V(\subtree{G}{r}{u})|)$.

\begin{figure}
\includegraphics[scale=1]{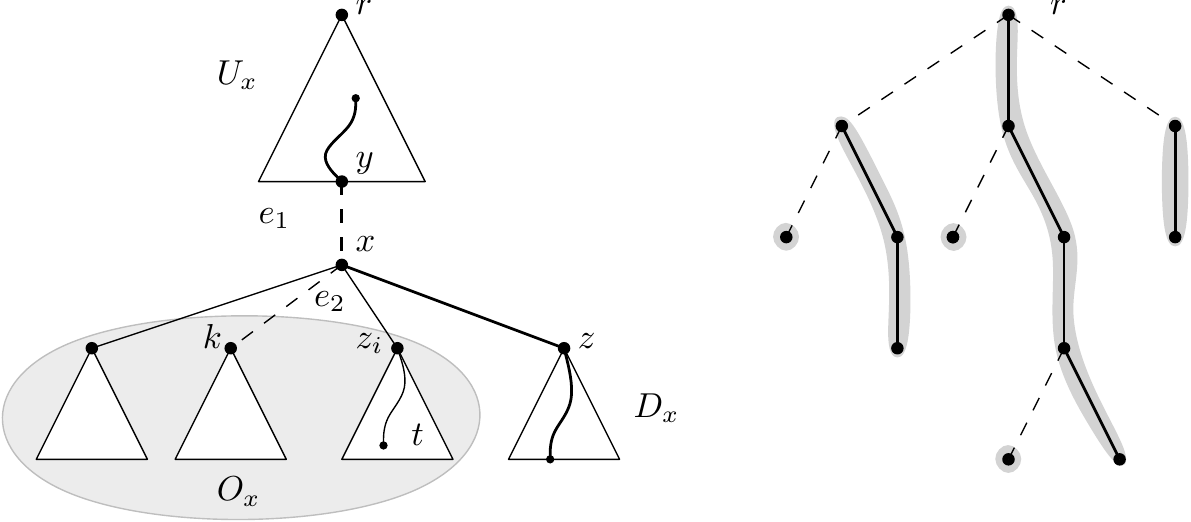}
\caption{Left: a view of the partition of the vertices induced by the removal of 
a
pair of edges of $E(\tree{G}{r})$.
Right: A path decompostion of a tree. Paths of the decomposition are
highlighted. Edges connecting the roots of the resulting subtrees to a path of
the decomposition are dashed.}
\label{fig:2_failure_structure}
\end{figure}

For the sake of the analysis, imagine the edges of paths considered by the
algorithm as if they were directed. Notice that no new edge entering a vertex
in $U_x$ can be added to $H$, as the shortest paths towards vertices in $U_x$
cannot change, and $H$ contains a shortest path tree $T$ of $G$. Hence, in the
following, we ignore all the edges entering vertices in $U_x$.

In Lines~\ref{ln:protecting_x_start}--\ref{ln:protecting_x_end}, the edges of
at most two paths are added to $H$. Moreover, by definition of
$\FirstLast(\cdot, \cdot)$, at most one edge of each path enters a vertex in
$D_x$.  This implies that the number of new edges is at most $O(O_x)$. In
Lines~\ref{ln:protecting_z_start}--\ref{ln:protecting_z_end}, at most $3$ paths
are considered as $\{\pi_{G-e_1}(z) \cap C(x)\}$ contains at most $2$ edges.
Each of those paths has at most one new edge which enters a vertex $q$ in $D_z$
since, once this happens, the shortest path from $q$ to $z$ of $T$ is already
in $H$. Again, the number of new edges is at most $O(O_x)$. In
Lines~\ref{ln:protecting_children_start}--\ref{ln:protecting_children_end}, at
most one edge for each children of $z_i \neq z$ of $x$ is added to $H$, and all
those children belong to $O_x$. Finally, in
Lines~\ref{ln:protecting_other_start}--\ref{ln:protecting_other_end} only new
edges entering vertices in $O_x$ are added to $H$, so their overall number is
$O(O_x)$.

As all the sets $O_x$ associated to the different vertices $x$ of $P$ are
pairwise vertex disjoint, we immediately have that at most
$O(|V(\subtree{G}{r}{v})|)$ edges are added to $H$ when path $P$ is examined.

The first path $P$ considered by Algorithm~\ref{alg:2_path} is the one obtained
by applying Lemma~\ref{lemma:path_decomposition_one_path} on $T$. The removal
of this path splits $T$ into a number $h$ of subtrees $T_1, \dots, T_h$ having
$\eta_1, \dots, \eta_h$ vertices respectively. Moreover we know that $\eta_i
\le \frac{n}{2} \; \forall i=1,\dots, h$ and that $\sum_{i=1}^h \eta_i \le n$.
If we reapply the procedure recursively, we get the following recurrence
equation describing the overall number of new edges:
\[
	S(n) = \sum_{i=1}^h S(\eta_i) + n
\]
which can be solved to show that $S(n) = O(n \log n)$.
To conclude the proof, we only need to notice that the set of paths
$\mathcal{P}$ used by Algorithm~\ref{alg:2_path} is defined exactly in this
very same recursive fashion, and that the tree $\hat{T}$ has $O(n)$ edges.
\end{proof}

Finally, we bound the running time of Algorithm~\ref{alg:2_path}:
\begin{lemma}
Algorithm~\ref{alg:2_path} requires $O(nm + n^2 \log n)$ time.
\label{lemma:2_path:complexity}
\end{lemma}
\begin{proof}
First of all, observe that a rough estimate of the time needed for computing 
the 
path decomposition $\mathcal{P}$ is $O(n^2)$ and that the time needed to build 
$\hat{T}$ is $O(n m)$ \cite{NPW03}. Moreover each vertex $x$ get considered at 
most once.
	
When the algorithm is considering a vertex $x$, a constant number of different
shortest paths are needed. Those can be computed in $O(m+n \log n)$ time
using the Dijkstra's algorithm where, for each vertex $v$, we also store the
last edge of its shortest path that (i) leaves the same connected component of
$r$ in $T-F$, (ii) leaves $\subtree{G}{r}{z}$, and (iii) enters the same 
connected
component as $v$ in $T-F$. This allows to implement $\FirstLast(\cdot)$ and to
add the edges needed in
Lines~\ref{ln:protecting_children_start}--\ref{ln:protecting_children_end},
\ref{ln:protecting_other_start}--\ref{ln:protecting_other_end} in time
proportional to the vertices in $O_x$. Hence, the overall time spent by adding
edges to $H$ is again $O(n \log n)$.
\end{proof}
\noindent By Lemmata~\ref{lemma:2_path_t_in_Tz}--\ref{lemma:2_path:size}, 
Theorem~\ref{thm:2_path_preproc} follows.
\section{Oracle Setting for $f=2$ and Proof of Theorem~\ref{thm:oracle_2_path}}
We here give a brief description of how to modify Algorithm~\ref{alg:2_path} in
order to build an oracle of size $O(n \cdot \log n)$ which is able to report,
with optimal query time, both a $3$-stretched shortest path in $G-F$ and its
distance, when $F$ contains two consecutive edges in $T$.

In order to do so, we first add an additional step to
Algorithm~\ref{alg:2_path} which computes an $O(n)$ size structure which is able 
to
answer LCA queries in $O(1)$ time \cite{HT84}. Then we store the tree $T$ and,
for each vertex $x$, its child $z$ on the path decomposition.

Whenever we are considering a vertex $x$ and its child $z \in P$, we also store
each path, say $\pi$, towards a vertex, say $u$, considered in
Lines~\ref{ln:protecting_x_start}--\ref{ln:protecting_x_end},
\ref{ln:protecting_z_start}--\ref{ln:protecting_z_end}, using a \emph{compact
representation}. To be more precise, let $s$ be the last vertex of $\pi$ which
belongs to the same component as $r$ in $T-F$, and let $q,q^\prime$ be the
first and last vertex of $\pi$ which belong to $T(z)$. We only store the (i)
vertices $s,q,q^\prime$, (ii) the subpaths $\pi[s:q]$, $\pi[q^\prime, u]$ along
with their lengths, and (iii) a reference to the position $x$ in the subpaths of
$\pi$, if any. If $q,q^\prime$ do not exists, we simply store $s$, $\pi[s:u]$,
$w(\pi[s,u])$, and a reference to $x$.

In Lines~\ref{ln:protecting_children_start}--\ref{ln:protecting_children_end},
we add one edge $(u,q)$ for each children $z_i \neq z$ of $x$. We store $(u,q)$
alongside $z_i$.

Finally, in
Lines~\ref{ln:protecting_other_start}--\ref{ln:protecting_other_end}, we add
some edges of the shortest path tree $\tree{G-x}{r}$. 
For each vertex $u \in O_x$, we store (i) the edge leading to its parent
in $\tree{G-x}{r}$, (ii) the last vertex $q$ of $pi_{G-x}(u)$ which is either in
$U$ or in $V(\subtree{G}{r}{z})$, (iii) the length of $\pi_{G-x}(u)[q,u]$, and 
iv)
the root of the subtree containing $u$ in $T-x$.

Since the amount of memory used to do so is always proportional to the vertices
in $O_x$ we have that the overall size
is still $O(n \log n)$. It is easy to see that, given a path 
failure\footnote{Once again, we focus on the failure of exactly two edges. To 
handle the failure of only one edge $e$, it suffices to store a single backup 
edge associated with $e$, as shown in \cite{NPW03}.} $F = \{ (y,x), (x,k) \}$
and a vertex $t$, we can answer a query by building (or computing the distance
of) $\pi^*(t)$ as described in the appropriate lemma in Lemmata
\ref{lemma:2_path_t_in_Tz}--\ref{lemma:2_path_t_not_in_Tz_x}. In order to do so
we need to know:
\begin{itemize}
\item The root of the subtrees of $T-x$ containing $t$.
\item Whether $\pi_{G-F}(t)$ contains $x$.
\end{itemize}
The former can be easily done by querying, in constant time, the least common 
ancestors of the pairs $t,z$ and $t,x$ in $T$ to determine if $z$ belongs to $U$ 
or $\subtree{G}{r}{z}$. If that is not the case, then the root of the sought 
subtree was explicitly stored and can be retrieved. As for the latter, we 
consider both cases. That is, we compute two candidate paths, we discard the one 
containing $(x,z_i)$, if any (this is done using the pointers to $x$), and we 
return the shortest of the remaining paths (or its distance). The above 
reasoning suffices to prove Theorem~\ref{thm:oracle_2_path}.

\end{document}